\documentclass[aps,showkeys]{revtex4}
\usepackage{braket}%为使用狄拉克符号而调用。
\usepackage{amsmath}%数学公式宏包
\usepackage{amssymb}%更多数学符号宏包，我们主要为使用数集的符号而调用。
\usepackage{amsthm}%定理环境包
\usepackage{graphicx}%使用插图宏包。
\usepackage{appendix}
\theoremstyle{definition}
\newtheorem{theorem}{Theorem}
\newtheorem{lemma}{Lemma}
\newtheorem{corollary}{Corollary}
\newtheorem{example}{Example}
\newtheorem{remark}{Remark}
\begin{document}
\title{Estimating the Schmidt numbers of quantum states via symmetric measurements}
\author{Hao-Fan Wang}
\email{2230502117@cnu.edu.cn}
\affiliation{School of Mathematical Sciences, Capital Normal University, Beijing 100048, China}
\author{Shao-Ming Fei}
\email{feishm@cnu.edu.cn}
\affiliation{School of Mathematical Sciences, Capital Normal University, Beijing 100048, China}
\begin{abstract}
The Schmidt numbers quantify the entanglement degree of quantum states. Quantum states with high Schmidt numbers provide a larger advantage in various quantum information processing tasks compared to quantum states with low Schmidt numbers. We derive a Schmidt number criterion based on the trace norm of the correlation matrix obtained from symmetric measurements. We show that our result is more effective than and superior to existing Schmidt number criteria by detailed examples.
\end{abstract}

\keywords{Schmidt number, Symmetric measurements, Quantum entanglement, Correlation matrix}

%style of springer
%\author*{\fnm{Haofan} \sur{Wang}}\email{2230502117@cnu.edu.cn}
%\author{\fnm{Shao-Ming} \sur{Fei}}\email{feishm@cnu.edu.cn}
%\affil{\orgdiv{School of Mathematical Sciences}, \orgname{Capital Normal University}, \orgaddress{\city{Beijing}, \postcode{100048}, \country{China}}}
%\abstract{The Schmidt numbers quantify the entanglement dimension of quantum states. We derive a Schmidt number criterion based on the trace norm of the correlation matrix obtained from symmetric measurements. We show that our Schmidt number criterion is more effective than and superior to existing criteria by detailed examples.}
%\keywords{Schmidt number, Symmetric measurements, Quantum entanglement, Correlation matrix}

\maketitle

\section{Introduction}
Entanglement is a fundamental resource for quantum information science and one of the key dividing factors between the quantum and the classical worlds. However, in practical applications, Neither determining whether a given state is entangled nor describing the entanglement degree of a given entangled state is an easy task.\par

A bipartite pure state $\ket{\psi}$ is called separable, if it is a product state $\ket{\psi}=\ket{u}\otimes\ket{v}$, otherwise it is entangled. More generally, a bipartite mixed state is separable if it can be written as a convex combination of pure product states,
\begin{equation*}
	\rho=\sum\limits_{i}p_{i}\ket{u_{i}}\bra{u_{i}}\otimes\ket{v_{i}}\bra{v_{i}}
\end{equation*}
where the $p_{i}$ form a probability distribution, so they are positive and sum up to one. A state that cannot be written in the above form is called entangled.\par

An important issue in the theory of quantum entanglement is the quantification and estimation of entanglement for composite systems. The Schmidt number is a well-known entanglement quantifier for bipartite states \cite{Terhal}, which introduces a hierarchy of entangled states and shows that a quantum state is separable if and only if its Schmidt number is one \cite{Sperling,Chruscinski}. Bipartite quantum states with higher Schmidt numbers are generally considered superior to those with lower Schmidt numbers in various information processing tasks. For example, it has been shown in Ref.\cite{Bae} that quantum states with higher Schmidt numbers have advantages in improving channel discrimination probability. \par

Despite the important advantages of entanglement states with higher Schmidt numbers in quantum information processing, a fundamental challenge is how to effectively detect the Schmidt number of a given quantum state. The first Schmidt number criterion is obtained by examining the fidelity between the quantum state and the maximally entangled states \cite{Terhal}. Then in Refs.\cite{Hulpke,Johnston} the authors presented a Schmidt number criterion by generalizing the well-known positive partial transpose (PPT) \cite{PPT1,PPT2} and the computable cross-norm or realignment criterion (CCNR) \cite{CCNR1,CCNR2} criteria. Later, the Schmidt number criteria based on Bloch decomposition \cite{Klockl} and covariance matrix \cite{Liu1,Liu2} have been derived. In addition, a witness-based method has also been
developed to detect the Schmidt number of a state \cite{Sanpera,Wyderka,Shi}. Recently, the authors in Ref.\cite{Tavakoli} proposed two elegant criteria to detect the Schmidt number based on the correlation matrix obtained from symmetric informationally complete measure (SIC POVM) and from mutually unbiased bases (MUBs). The results are generalized to the one based on general SIC POVM (GSIC POVM) \cite{ZWang}, as a generalization of the entanglement criterion given in Ref.\cite{Lai}.

It is well-known that mutually unbiased measurements (MUMs) \cite{Kalev1} and GSIC POVM \cite{Kalev2} are the natural extensions of MUBs and SIC POVM, respectively. In 2022, the symmetric measurement or $(N,M)$-POVM has been proposed \cite{Siudzinska}, which includes MUMs and GSIC POVM as special cases. In fact, as the generalizations of symmetric measurement, the concepts of generalized symmetric measurements \cite{Siudzinska11} and generalized equiangular measurements \cite{Siudzinska12} have been proposed recently.

In this work, we derive a criterion of Schmidt number detection based on the trace norm of the correlation matrix whose entries are obtained via symmetric measurements. It can be considered as a generalization of the GSIC criterion \cite{ZWang} and the MUB criterion \cite{Tavakoli}. We show that our criterion is more efficient in detecting the Schmidt number than some existing criteria by detailed examples. Moreover, from the proof of our criterion of Schmidt number, we also obtain a class of lower bounds of concurrence induced by symmetric measurements, which generalizes the results given in Ref.\cite{Haofan}.

\section{Preliminaries}
We first recall the definition and some properties of $(N,M)$-POVM \cite{Siudzinska}. A set of $N$ $d$-dimensional POVMs $\{E_{\alpha,k}|k=1,2\cdots,M\}$ ($\alpha=1,2,\cdots,N$) constitute an $(N,M)$-POVM if
\begin{flalign*}
	{\rm tr}(E_{\alpha,k}) &= \dfrac{d}{M}, \\
	{\rm tr}(E_{\alpha,k}^{2}) &= x,\\
	{\rm tr}(E_{\alpha,k}E_{\alpha,l}) &= \dfrac{d-Mx}{M(M-1)},~~ l\neq k\\
	{\rm tr}(E_{\alpha,k}E_{\beta,l}) &= \dfrac{d}{M^{2}},~~ \beta\neq\alpha
\end{flalign*}
where the parameter $x$ satisfies $\dfrac{d}{M^{2}}<x\leq \min\left\{\dfrac{d^{2}}{M^{2}},\dfrac{d}{M}\right\}$. When $N(M-1)=d^{2}-1$, the $(N,M)$-POVM is called an informationally complete $(N,M)$-POVM.
For any finite dimension $d$ ($d>2$), there exist at least four different types of informationally complete $(N,M)$-POVM: (1) $N=1$ and $M=d^{2}$ (GSIC POVM), (2) $N=d+1$ and $M=d$ (MUMs), (3) $N=d^{2}-1$ and $M=2$, (4) $N=d-1$ and $M=d+2$.

From orthonormal Hermitian operator basis $\{G_{0}=I_{d}/\sqrt{d},\,G_{\alpha,k}|\alpha=1,\cdots,N;\,k=1,\cdots,M-1\}$ with ${\rm tr}(G_{\alpha,k})=0$, an informationally complete $(N,M)$-POVM is given by
\begin{equation*}
	E_{\alpha,k}=\dfrac{1}{M}I_{d}+tH_{\alpha,k},
\end{equation*}
where
\begin{equation*}
	H_{\alpha,k}=\begin{cases}
		G_{\alpha}-\sqrt{M}(\sqrt{M}+1)G_{\alpha,k},~ k=1,\cdots,M-1\\
		(\sqrt{M}+1)G_{\alpha}, k=M
	\end{cases}
\end{equation*}
with $G_{\alpha}=\sum\limits_{k=1}^{M-1}G_{\alpha,k}$. The parameter $t$ should be chosen such that $E_{\alpha,k}\geq 0$, which is equivalent to
\begin{equation*}
	-\dfrac{1}{M}\dfrac{1}{\lambda_{\max}}\leq t \leq \dfrac{1}{M}\dfrac{1}{|\lambda_{\min}|},
\end{equation*}
where $\lambda_{\max}$ and $\lambda_{\min}$ are the minimal and maximal eigenvalue
from among all eigenvalues of $H_{\alpha,k}$, respectively. The parameters $t$ and $x$ satisfy the following relation,
\begin{equation*}
	x=\dfrac{d}{M^{2}}+t^{2}(M-1)(\sqrt{M}+1)^{2}.
\end{equation*}

Next, we recall the definition of Schmidt number. A bipartite pure state $\ket{\psi}\in \mathcal{H}_{A}\otimes\mathcal{H}_{B}$ has a Schmidt decomposition $ \ket{\psi}=\sum\limits_{i=1}^{r} \lambda_i \ket{e_{i}}\otimes\ket{f_{i}}$, where $\lambda_i>0$ and  $\sum\limits_{i=1}^{r}\lambda_i^2=1$, $\{\ket{e_i}\}$ and $\{\ket{f_i}\}$ are the orthonormal bases in $\mathcal{H}_{A}$ and $\mathcal{H}_{B}$, respectively. The number $r$ is called the Schmidt rank of $\ket{\psi}$, denoted as ${\rm SR}(\ket{\psi})$\cite{Ekert}. The Schmidt number of a bipartite mixed state $\rho$ in $\mathcal{H}_{A}\otimes\mathcal{H}_{B}$ is defined as \cite{Terhal}
\begin{equation*}
	{\rm SN}(\rho) = \min\limits_{\{p_{i},\ket{\psi_{i}}\}}\max\limits_{i}{\rm SR}(\ket{\psi_{i}}),
\end{equation*}
where the minimization goes over all possible pure state decompositions of $\rho=\sum_i p_i \ket{\psi_i}\bra{\psi_i}$. Obviously, the Schmidt numbers of a pure state is its Schmidt rank.

\section{Symmetric measurement based Schmidt number criterion}
To derive Schmidt number criterion for symmetric measurements, we first need to prove the following lemma, which generalizes the Lemma 1 in Ref.\cite{Siudzinska} to non-Hermitian matrices.

\begin{lemma}
	Let $\{E_{\alpha,k}|\alpha=1,\cdots,N;\,k=1,\cdots,M\}$ be an informationally complete $(N, M)$-POVM on $d$ dimensional Hilbert space $\mathcal{H}$ with free parameter $x$. Then for any linear operator $\sigma\in L(\mathcal{H})$, we have
	\begin{equation*}\label{lemma1}
		\sum_{\alpha=1}^{N}\sum_{k=1}^{M}|{\rm tr}(E_{\alpha,k}\sigma)|^{2}=\dfrac{d(M^{2}x-d){\rm tr}(\sigma\sigma^{\dagger})+(d^{3}-M^{2}x)|{\rm tr}(\sigma)|^{2}}{dM(M-1)}.
	\end{equation*}
\end{lemma}

\begin{proof}
Denote $w=\dfrac{d}{M}$, $y=\dfrac{d-Mx}{M(M-1)}$ and $z=\dfrac{d}{M^{2}}$. For any linear operator $\sigma\in L(\mathcal{H})$, we can verify that
	\begin{equation*}
		\sigma = \sum\limits_{\alpha=1}^{N}\sum\limits_{k=1}^{M}{\rm tr}(E_{\alpha,k}\sigma)F_{\alpha,k},
	\end{equation*}
    where $F_{\alpha,k}=\dfrac{1}{x-y}\left(E_{\alpha,k}-\dfrac{(N-1)z+y}{Nw}I_{d}\right)$. Define $p_{\alpha,k}={\rm tr}(E_{\alpha,k}\sigma)$. Then
    \begin{equation*}
    	{\rm tr}(\sigma\sigma^{\dagger})=\sum\limits_{\alpha,\beta=1}^{N}
    \sum\limits_{k,l=1}^{M}p_{\alpha,k}\overline{p_{\beta,l}}{\rm tr}(F_{\alpha,k}F_{\beta,l}).
    \end{equation*}
By definition it follows that
    \begin{equation*}
    	{\rm tr}(F_{\alpha,k}F_{\beta,l})=\dfrac{1}{(x-y)^{2}}\left[{\rm tr}(E_{\alpha,k}E_{\beta,l})+A(dA-2w)\right],
    \end{equation*}
    where $A=\dfrac{(N-1)z+y}{Nw}$. As a result, we have
    \begin{flalign*}
    	& \begin{array}{ll}
    		{\rm tr}(\sigma\sigma^{\dagger})
    		& = \dfrac{1}{(x-y)^{2}}\Bigg[\sum\limits_{\alpha=1}^{N}\sum\limits_{k=1}^{M}|p_{\alpha,k}|^{2}{\rm tr}(E_{\alpha,k}^2)+\sum\limits_{\alpha=1}^{N}\sum\limits_{\substack{k,l=1 \\ l\neq k}}^{M}
    		p_{\alpha,k}\overline{p_{\alpha,l}}{\rm tr}(E_{\alpha,k}E_{\alpha,l})\\
    		& \qquad\qquad\quad +\sum\limits_{\substack{\alpha,\beta=1 \\ \beta\neq\alpha}}^{N} \sum\limits_{k,l=1}^{M}p_{\alpha,k}\overline{p_{\beta,l}}{\rm tr}(E_{\alpha,k}E_{\beta,l})+N^{2}A(dA-2w)|{\rm tr}(\sigma)|^{2}\Bigg]\\
    		& = \dfrac{1}{(x-y)^{2}}\left[Cx+Ny|{\rm tr}(\sigma)|^{2}-Cy+N(N-1)z|{\rm tr}(\sigma)|^{2}+N^2A(dA-2w)|{\rm tr}(\sigma)|^{2}\right]\\
    		& = \dfrac{1}{(x-y)^{2}}\left[(x-y)C+N(y+(N-1)z)|{\rm tr}(\sigma)|^{2}+N^2A(dA-2w)|{\rm tr}(\sigma)|^{2}\right]\\
    		& = \dfrac{1}{(x-y)^{2}}\left[(x-y)C+\left(N+\dfrac{(dA-2w)N}{w}\right)(y+(N-1)z)|{\rm tr}(\sigma)|^{2}\right]\\
    		& =\dfrac{dM(M-1)C+(M^2x-d^{3})|{\rm tr}(\sigma)|^{2}}{d(M^2x-d)},
    	\end{array}&
    \end{flalign*}
    where $C=\sum\limits_{\alpha=1}^{N}\sum\limits_{k=1}^{M}|p_{\alpha,k}|^{2}$ and we have used $(M-1)N=d^{2}-1$ in the last equality. The proof of the lemma is completed by adjusting the above equation.
\end{proof}

Now, we can prove our main result following the methodology and notation in Ref.\cite{Tavakoli}.

\begin{theorem}
Let $\{E_{\alpha,k}^{A}|\alpha=1,\cdots,N_{A};\,k=1,\cdots,M_{A}\}$ be an informationally complete $(N_{A}, M_{A})$-POVM on Hilbert space $\mathcal{H}_{A}$ and $\{E_{\beta,l}^{B}|\beta=1,\cdots,N_{B};\,l=1,\cdots,M_{B}\}$ an informationally complete $(N_{B}, M_{B})$-POVM on Hilbert space $\mathcal{H}_{B}$. Denote $p_{\alpha,k;\beta,l}={\rm tr}\left(\rho_{AB}\left(E_{\alpha,k}^{A}\otimes E_{\beta,l}^{B}\right)\right)$ and $\mathcal{P}(\rho_{AB})=\left(p_{\alpha,k;\beta,l}\right)_{N_{A}M_{A}\times N_{B}M_{B}}$
for a bipartite state $\rho_{AB}$ in $\mathcal{H}_{A}\otimes\mathcal{H}_{B}$. If the Schmidt number of $\rho_{AB}$ is at most $r$, it holds that
	\begin{equation*}
		\|\mathcal{P}(\rho_{AB})\|_{\rm tr}\leq\dfrac{L}{K}+(r-1)\dfrac{R}{K},
	\end{equation*}
	where
	\begin{eqnarray}
		K&=&\sqrt{d_{A}d_{B}(M_{A}-1)(M_{B}-1)},\nonumber\\[2mm]
		L&=&\sqrt{\dfrac{(d_{A}-1)(d_{B}-1)(M_{A}^{2}x_{A}
+d_{A}^{2})(M_{B}^{2}x_{B}+d_{B}^{2})}{M_{A}M_{B}}},\nonumber\\[2mm]
		 R&=&\sqrt{\dfrac{d_{A}d_{B}(M_{A}^{2}x_{A}-d_{A})(M_{B}^{2}x_{B}-d_{B})}{M_{A}M_{B}}}.\nonumber
	\end{eqnarray}
\end{theorem}

\begin{proof}
Since the trace norm is convex, without loss of generality, instead of mixed states we consider pure states with Schmidt rank $r$, $\ket{\psi}=\sum\limits_{s=1}^{r}\lambda_s\ket{ss}$, where $\{\lambda_s\}$ is the set of Schmidt coefficients with $\sum\limits_{s=1}^r\lambda_s^2=1$. As a result, we obtain
	\begin{flalign}
		& \begin{array}{ll}
			{\rm tr}\left(\ket{\psi}\bra{\psi}\left(E_{\alpha,k}^{A}\otimes E_{\beta,l}^{B}\right)\right)
		    & = \sum\limits_{s,t=1}^r\lambda_s\lambda_t\braket{ss|E_{\alpha,k}^{A}\otimes E_{\beta,l}^{B}|tt}\\
			& =\sum\limits_{s=1}^{r}\lambda_{s}^{2}\braket{ss|E_{\alpha,k}^{A}\otimes E_{\beta,l}^{B}|ss}+\sum\limits_{s\neq t}\lambda_s\lambda_t\braket{ss|E_{\alpha,k}^{A}\otimes E_{\beta,l}^{B}|tt}.
		\end{array}&\label{pab}
	\end{flalign}
Define $[D_s]_{\alpha,k;\beta,l}=\braket{ss|E_{\alpha,k}^{A}\otimes E_{\beta,l}^{B}|ss}$ and $[O_{s,t}]_{\alpha,k;\beta,l}=\braket{ss|E_{\alpha,k}^{A}\otimes E_{\beta,l}^{B}|tt}$. It then follows that
    \begin{eqnarray}
    	\|D_s\|_{\rm tr}&=&\sqrt{\sum\limits_{\alpha=1}^{N_{A}}\sum\limits_{k=1}^{M_{A}}
    \braket{s|E_{\alpha,k}|s}^{2}}\sqrt{\sum\limits_{\beta=1}^{N_{B}}
    \sum\limits_{l=1}^{M_{B}}\braket{s|E_{\beta,l}|s}^{2}},\nonumber\\
    	\|O_{s,t}\|_{\rm tr}&=&\sqrt{\sum\limits_{\alpha=1}^{N_{A}}\sum\limits_{k=1}^{M_{A}}
    |\braket{s|E_{\alpha,k}|t}|^{2}}\sqrt{\sum\limits_{\beta=1}^{N_{B}}
    \sum\limits_{l=1}^{M_{B}}|\braket{s|E_{\beta,l}|t}|^{2}},~~~s\neq t.\nonumber
    \end{eqnarray}
Using Lemma 1, we obtain
    \begin{eqnarray}
    	\|D_s\|_{\rm tr}&=&\sqrt{\dfrac{(d_{A}-1)(M_{A}^{2}x_{A}+d_{A}^{2})}{d_{A}M_{A}(M_{A}-1)}}
    \sqrt{\dfrac{(d_{B}-1)(M_{B}^{2}x_{B}+d_{B}^{2})}{d_{B}M_{B}(M_{B}-1)}},\nonumber\\
    	\|O_{s,t}\|_{\rm tr}&=&\sqrt{\dfrac{M_{A}^{2}x_{A}-d_{A}}{M_{A}(M_{A}-1)}}
    \sqrt{\dfrac{M_{B}^{2}x_{B}-d_{B}}{M_{B}(M_{B}-1)}},~~~s\neq t.\nonumber
    \end{eqnarray}
    From Eq. (\ref{pab}), we obtain
    \begin{flalign}
    	& \begin{array}{ll}
    		\|\mathcal{P}(\ket{\psi}\bra{\psi})\|_{\rm tr}
    		& \leq \sum\limits_{s=1}^{r}\lambda_{s}^{2}\|D_s\|_{\rm tr}+\sum\limits_{s\neq t}\lambda_s\lambda_t\|O_{s,t}\|_{\rm tr}\\
    		& =\dfrac{1}{K}\left[\left(\sum\limits_{s=1}^{r}\lambda_{s}^{2}\right)(L-R)+\left(\sum\limits_{s=1}^{r}\lambda_{s}\right)^{2}R\right]\\
    		& =\dfrac{L}{K}+\dfrac{R}{K}\left[\left(\sum\limits_{s=1}^{r}\lambda_{s}\right)^{2}-1\right].
    	\end{array}&\label{key}
    \end{flalign}
    The proof of Theorem 1 is now complete by using the fact that $\left(\sum\limits_{s=1}^{r}\lambda_{s}\right)^{2}\leq r$\cite{Terhal}.
\end{proof}

\begin{remark}
	If $\rho_{AB}$ is a separable, then ${\rm SN}(\rho_{AB})=1$. Consequently from Theorem 1 we have $\|\mathcal{P}(\rho_{AB})\|_{\rm tr}\leq\sqrt{\frac{(d_{A}-1)(M_{A}^{2}x_{A}+d_{A}^{2})}{d_{A}M_{A}(M_{A}-1)}}
	\sqrt{\frac{(d_{B}-1)(M_{B}^{2}x_{B}+d_{B}^{2})}{d_{B}M_{B}(M_{B}-1)}}$. Hence our Theorem 1 covers the entanglement criterion given in Ref. \cite{Siudzinska,Tang}.
\end{remark}

\begin{remark}
	Set $M_{A}=d_{A}^{2}$, $M_{B}=d_{B}^{2}$ and $N_{A}=N_{B}=1$ in Theorem 1. Then
	\begin{eqnarray}
		K&=&\sqrt{d_{A}d_{B}(d_{A}^{2}-1)(d_{B}^{2}-1)},\nonumber\\[2mm]
		L&=&\sqrt{(d_{A}-1)(d_{B}-1)(x_{A}d_{A}^{2}+1)(x_{B}d_{B}^{2}+1)},\nonumber\\[2mm]
		R&=&\sqrt{(x_{A}d_{A}^{3}-1)(x_{B}d_{B}^{3}-1)}.\nonumber
	\end{eqnarray}
	Hence, Theorem 1 generalizes the GSIC Schmidt number criterion proposed in Ref.\cite{ZWang}. Similarly, take $N_{A}-1=M_{A}=d_{A}$, $N_{B}-1=M_{B}=d_{B}$ and $x_{A}=x_{B}=1$. We obtain $K=\frac{1}{2}L=R=\sqrt{d_{A}d_{B}(d_{A}-1)(d_{B}-1)}$. Then $\|\mathcal{P}(\rho_{AB})\|_{\rm tr}\leq 1+r$. This indicates that Theorem 1 also generalizes the MUBs Schmidt number criterion given in \cite{Tavakoli}.
\end{remark}

\begin{corollary}
Let $\{E_{\alpha,k}|\alpha=1,\cdots,N;\,k=1,\cdots,M\}$ be an informationally complete $(N, M)$-POVM on Hilbert space $\mathcal{H}$. Denote $p_{\alpha,k;\beta,l}={\rm tr}\left(\rho\left(E_{\alpha,k}\otimes E_{\beta,l}\right)\right)$ and $\mathcal{P}(\rho)=\left(p_{\alpha,k;\beta,l}\right)_{d^{2}\times d^{2}}$ $\rho$ for a bipartite state in $\mathcal{H}\otimes\mathcal{H}$. If the Schmidt number of $\rho$ is at most $r$, it holds that
	\begin{equation*}
		\|\mathcal{P}(\rho)\|_{\rm tr}\leq \dfrac{(d-1)(M^{2}x+d^{2})}{dM(M-1)}+(r-1)\dfrac{M^{2}x-d}{M(M-1)}.
	\end{equation*}
\end{corollary}

Actually, we can use another way to prove Corollary 1. Let $\ket{\psi}=\sum\limits_{i=1}^{r}\lambda_{i}\ket{ii}$ be a pure state in $\mathcal{H}\otimes\mathcal{H}$. Since $\left(\sum\limits_{i=1}^{r}\lambda_{i}\right)^{2}\leq r$ and $\sum\limits_{i=1}^{r}\lambda_{i}^{2}=1$, we have $r-1\geq 2\sum\limits_{i< j}\lambda_{i}\lambda_{j}$. Due to every informationally complete $(N, M)$-POVM is a conical 2-design \cite{Siudzinska2,Huang}, we can follow the method in Ref.\cite{Haofan} to obtain 
\begin{equation*}
	\dfrac{M(M-1)}{xM^{2}-d}\left(\|\mathcal{P}(\ket{\psi}\bra{\psi})\|_{\rm tr}-\dfrac{(d-1)(xM^{2}+d^{2})}{dM(M-1)}\right)=2\sum\limits_{i< j}\lambda_{i}\lambda_{j}\leq r-1.
\end{equation*}
This indicates that we have proven Corollary 1 in another way. It is easy to discover that this proof method is applicable to any conical 2-design. Therefore, this is very meaningful observation, which implies that we can propose new Schmidt number criterion based on other new conical 2-designs.

In fact, we can obtain a lower bound of concurrence based on the proof of Theorem 1. The concurrence of a bipartite state $\ket{\psi}$ is defined by
$C(\ket{\psi})=\sqrt{2(1-{\rm tr}(\rho_{A}^{2}))}$\cite{Rungta,Chen}, where $\rho_{A}={\rm tr}_{B}(\ket{\psi}\bra{\psi})$ is the reduced state obtained by tracing over the subsystem $B$. The concurrence of a mixed state $\rho_{AB}$ is given by the convex roof extension,
$C(\rho_{AB})=\min\limits_{\{p_{i},\ket{\psi_{i}}\}}\sum\limits_{i}p_{i}C(\ket{\psi_{i}})$,
where the minimum is taken over all possible pure state decompositions of $\rho_{AB} = \sum\limits_{i}p_{i}\ket{\psi_{i}}\bra{\psi_{i}}$, with $p_{i}\geq 0$ and $\sum\limits_{i}p_{i}=1$. Denote $d=\min\{d_{A},d_{B}\}$. According to Eq. (\ref{key}), we obtain
\begin{equation*}
	\frac{K}{R}\left(\|\mathcal{P}(\ket{\psi}\bra{\psi})\|_{\rm tr}-\frac{L}{K}\right)\leq \left(\sum\limits_{i}\lambda_{i}\right)^{2}-1=2\sum\limits_{i<j}\lambda_{i}\lambda_{j},
\end{equation*}
from which we can prove that
\begin{equation}\label{remark3}
	C(\rho_{AB})\geq \frac{K}{R}\sqrt{\frac{2}{d(d-1)}}\left(\|\mathcal{P}(\rho_{AB})\|_{\rm tr}-\frac{L}{K}\right)
\end{equation}
by using $C(\ket{\psi})\geq 2\sqrt{\dfrac{2}{d(d-1)}}\sum\limits_{i<j}\lambda_{i}\lambda_{j}$ \cite{Chen} and the convex property of the trace norm. When $\mathcal{H}_{A}=\mathcal{H}_{B}$ and $(N_{A},M_{A})=(N_{B},M_{B})$, (\ref{remark3}) reduces to the lower bound of concurrence given by Theorem 1 in Ref.\cite{Haofan}.

\section{Examples}
Let us consider several examples to illustrate our conclusions.
\begin{example}
	Consider the following $2\otimes 4$ state,
	$\rho(\tau,q)=q\ket{\xi}\bra{\xi}+(1-q)\rho_{\tau}$, where $\rho_{\tau}$ is the bound entangled state proposed by Horodecki \cite{state},
	\begin{equation*}
		\rho_{\tau}=\dfrac{1}{1+7\tau}
		\begin{pmatrix}
			\tau & 0 & 0 & 0 & 0 & \tau & 0 & 0 \\
			0 & \tau & 0 & 0 & 0 & 0 & \tau & 0 \\
			0 & 0 & \tau & 0 & 0 & 0 & 0 & \tau \\
			0 & 0 & 0 & \tau & 0 & 0 & 0 & 0 \\
			0 & 0 & 0 & 0 & \frac{1+\tau}{2} & 0 & 0 & \frac{\sqrt{1-\tau^{2}}}{2} \\
			\tau & 0 & 0 & 0 & 0 & \tau & 0 & 0 \\
			0 & \tau & 0 & 0 & 0 & 0 & \tau & 0 \\
			0 & 0 & \tau & 0 & \frac{\sqrt{1-\tau^{2}}}{2} & 0 & 0 & \frac{1+\tau}{2}
		\end{pmatrix}
	\end{equation*}
	with $0<\tau<1$ and $\ket{\xi}=\frac{1}{\sqrt{2}}(\ket{00}+\ket{11})$. We construct a $(3,2)$-POVM with the Hermitian basis operators given by Pauli matrices
	$G_{11}=\dfrac{1}{\sqrt{2}}\begin{pmatrix}
		0 & 1\\
		1 & 0\\
	\end{pmatrix}$,~
	$G_{21}=\dfrac{1}{\sqrt{2}}\begin{pmatrix}
		0 & -\mathrm{i}\\
		\mathrm{i} & 0\\
	\end{pmatrix}$,~
	$G_{31}=\dfrac{1}{\sqrt{2}}\begin{pmatrix}
		1 & 0\\
		0 & -1\\
	\end{pmatrix}$,
	and a $(5,4)$-POVM with the Hermitian basis operators $G_{\alpha,k}$ given in Appendix A. It is verified that the corresponding parameters $x_1=\frac{1}{2}+(\sqrt{2}+1)^{2}t_{1}^{2}$ with $t_{1}\in [-0.2929,0.2929]$ and $x_{2}=\frac{1}{4}+27t_{2}^{2}$ with $t_{2}\in [-0.0572,0.0680]$. We take $\rho_{q}=\rho(0.9,q)$. In Fig.1, the red curve is the lower bound of ${\rm SN}(\rho_{q})-1$ from our Theorem 1 based on $(3,2)$-POVM and $(5,4)$-POVM with $t_{1}=t_{2}=0.01$. In other words, the red curve is the value of $\frac{K}{R}(\|\mathcal{P}(\rho_{q})\|_{\rm tr}-\frac{L}{K})$ based on $(3,2)$-POVM and $(5,4)$-POVM with $t_{1}=t_{2}=0.01$. The green curve is the lower bound of ${\rm SN}(\rho_{q})-1$ from Theorem 1 in Ref.\cite{ZWang} based on GSIC POVMs with $a_{1}=0.1277$ and $a_{2}=0.04984$. In other words, the green curve is the value of $\frac{K}{R}(\|\mathcal{P}(\rho_{q})\|_{\rm tr}-\frac{L}{K})$ based on GSIC POVMs with $a_{1}=0.1277$ and $a_{2}=0.04984$. In fact, the red curve shows that $\rho(0.9,q)$ is entangled for $0.42115\leq q\leq 1$, and the green curve shows that Theorem 1 in Ref.\cite{ZWang} fails to detect the entanglement of $\rho(0.9,q)$. This indicates that our Schmidt number criterion detect more entangled states than the GSIC Schmidt number criterion. In other words, our Schmidt number criterion detect more quantum states are having Schmidt number strictly greater than 1. Thus, our criterion is more efficient in detecting the Schmidt number than the GSIC criterion introduced in Ref.\cite{ZWang}.
	\begin{figure}[t]
		\includegraphics[width=10cm]{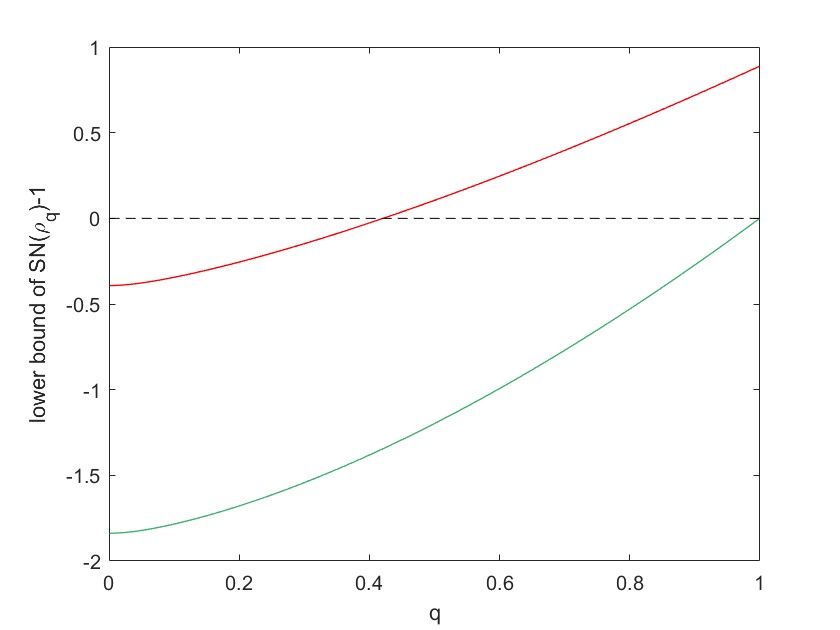}
		\caption{Red curve is the lower bound of ${\rm SN}(\rho_{q})-1$ from our Theorem 1 based on $(3,2)$-POVM and $(5,4)$-POVM. Green curve is the lower bound of ${\rm SN}(\rho_{q})-1$ from Theorem 1 in Ref.\cite{ZWang} based on GSIC POVMs with $a_{1}=0.1277$ and $a_{2}=0.04984$.}
	\end{figure}
\end{example}

\begin{example}
	Ref.\cite{Weilenmann} introduces the following mixed two-ququart state
	\begin{equation*}
		 \rho=\dfrac{1}{2}\ket{\phi^{+}_{3}}\bra{\phi^{+}_{3}}+\dfrac{1}{4}(\ket{23}+\ket{32})(\bra{23}+\bra{32})
	\end{equation*}
	where $\ket{\phi^{+}_{3}}=\dfrac{1}{\sqrt{3}}(\ket{00}+\ket{11}+\ket{33})$. Now, we consider the mixture of $\rho$ and pure state $\ket{\xi}=\dfrac{1}{5}\ket{00}+\dfrac{1}{5}\ket{11}+\dfrac{\sqrt{23}}{5}\ket{22}$,
	\begin{equation*}
		\rho_{p}=p\rho+(1-p)\ket{\xi}\bra{\xi}
	\end{equation*}
	Take the $(N,M)$-POVM in Corollary 1 to be $(5,4)$-POVM with the Hermitian basis operator $G_{\alpha, k}$ given in Appendix A.  It is verified that the parameter $x=\frac{1}{4}+27t^{2}$ with $t\in [-0.0572,0.0680]$. In Fig. 2, the red curve is the lower bound of ${\rm SN}(\rho_{p})-1$ from our Corollary 1 based on $(5,4)$-POVM with $t=0.01$, which shows that $\rho_{p}$ is a entanglement state with Schmidt numbers strictly greater than 2 for $0.5219\leq p\leq 1$. In other words, the red curve is the value of $\frac{M(M-1)}{xM^{2}-d}\left(\|\mathcal{P}(\rho_{p})\|_{\rm tr}-\frac{(d-1)(xM^{2}+d^{2})}{dM(M-1)}\right)$ based on $(5,4)$-POVM with $t=0.01$. The orange curve is the lower bound of ${\rm SN}(\rho_{p})-1$ from realignment \cite{Hulpke,Johnston}, which shows that $\rho_{p}$ is a entanglement state with Schmidt numbers strictly greater than 2 for $0.5475\leq p\leq 1$. In other words, orange curve is the value of $\|\mathcal{R}(\rho_{p})\|_{\rm tr}-1$, where $\mathcal{R}(\rho_{p})$ is realigned matrix of $\rho_{p}$. Compared with realignment criterion, our Schmidt number criterion detect more entangled states are having Schmidt number strictly greater than 2. Therefore, our criterion is more efficient in detecting the Schmidt number than realignment criterion.
	\begin{figure}[t]
		\includegraphics[width=10cm]{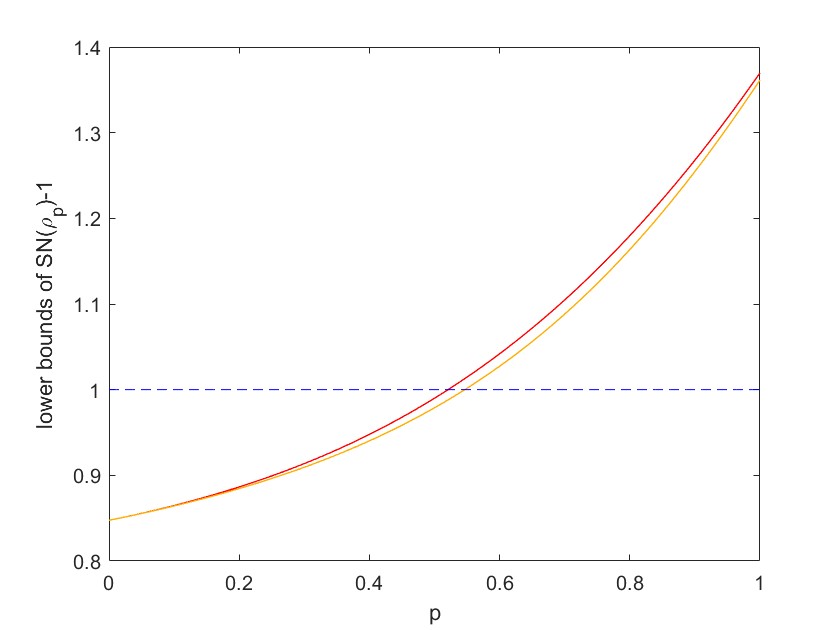}
		\caption{Red curve is the lower bound of ${\rm SN}(\rho_{p})-1$ from our Corollary 1 based on $(5,4)$-POVM with $t=0.01$. Orange curve is the lower bound of ${\rm SN}(\rho_{p})-1$ from realignment.}
	\end{figure}
\end{example}

\begin{example}
	Consider the following isotropic state \cite{isostate},
	\begin{equation*}    	
		\rho_{v}=v\ket{\Psi^{+}}\bra{\Psi^{+}}+(1-v)\dfrac{I_{d}}{d^{2}},
	\end{equation*}
	where $I$ is the identity operator on ${H\otimes H}$, $\ket{\Psi^{+}}=\dfrac{1}{\sqrt{d}}\sum\limits_{i=1}^{d}\ket{ii}$ and $0<v<1$. It is easy to know that
	\begin{equation*}
		\|\mathcal{P}(\rho_{v})\|_{\rm tr}=\dfrac{N}{M}+\dfrac{vN(M^{2}x-d)}{dM}=\dfrac{(d-1)(M^{2}x+d^{2})}{dM(M-1)}+\dfrac{N(M^{2}x-d)}{dM}\left(v-\dfrac{1}{d+1}\right)
	\end{equation*}
	by direct calculation \cite{Haofan}. In Ref.\cite{Terhal}, the authors actually use the fidelity criterion (i.e. Lemma 1 in Ref.\cite{Terhal}) to detect the Schmidt number of state $\rho_{v}$. The fidelity criterion says that the Schmidt number of state $\rho_v$ is at least $r+1$ if $v$ greater than the critical value $v_{\rm opt}=\frac{rd-1}{d^2-1}$. For any value of $v>v_{\rm opt}$ it holds that
	\begin{flalign*}
		& \begin{array}{ll}
			\|\mathcal{P}(\rho_{v})\|_{\rm tr}
			& >\dfrac{(d-1)(M^{2}x+d^{2})}{dM(M-1)}+\dfrac{N(M^{2}x-d)}{dM}\left(\dfrac{rd-1}{d^2-1}-\dfrac{1}{d+1}\right)\\
			& =\dfrac{(d-1)(M^{2}x+d^{2})}{dM(M-1)}+(r-1)\dfrac{M^{2}x-d}{M(M-1)}.
		\end{array}&
	\end{flalign*}
    From this and Corollary 1, we know that the Schmidt number of state $\rho_v$ is at least $r+1$. Thus, our criterion must not be weaker than the fidelity criterion for the state $\rho_v$.
\end{example}

\begin{example}
	Consider the mixture of the bound entangled state considered in Ref.\cite{state},
	\begin{equation*}
		\rho_{\tau}=\dfrac{1}{1+8\tau}
		\begin{pmatrix}
			\tau & 0 & 0 & 0 & \tau & 0 & 0 & 0 & \tau\\
			0 & \tau & 0 & 0 & 0 & 0 & 0 & 0 & 0\\
			0 & 0 & \tau & 0 & 0 & 0 & 0 & 0 & 0\\
			0 & 0 & 0 & \tau & 0 & 0 & 0 & 0 & 0\\
			\tau & 0 & 0 & 0 & \tau & 0 & 0 & 0 & \tau\\
			0 & 0 & 0 & 0 & 0 & \tau & 0 & 0 & 0\\
			0 & 0 & 0 & 0 & 0 & 0 & \frac{1+\tau}{2} & 0 & \frac{\sqrt{1-\tau^{2}}}{2}\\
			0 & 0 & 0 & 0 & 0 & 0 & 0 & 0 & 0\\
			\tau & 0 & 0 & 0 & \tau & 0 & \frac{\sqrt{1-\tau^{2}}}{2} & 0 & \frac{1+\tau}{2}
		\end{pmatrix}
	\end{equation*}
	and the $9\times 9$ identity matrix $I_{9}$,
	\begin{equation*}
		\rho(\tau,q)=q\rho_{\tau}+\frac{1-q}{9}I_{9}.
	\end{equation*}
    Take the $(N,M)$-POVM in Corollary 1 to be $(8,2)$-POVM with the Hermitian basis operator $G_{\alpha, k}$ given in Appendix B. It is verified that the parameter $x=\frac{3}{4}+t^{2}(\sqrt{2}+1)^{2}$ with $t\in[-0.2536,0.2536]$. Consider $\rho=\rho(\tau,0.995)$. In Fig. 3, the red curve is the lower bound of ${\rm SN}(\rho)-1$ from our Corollary 1 based on $(8,2)$-POVM with $t=0.01$. In other words, the red curve is the value of $\frac{M(M-1)}{xM^{2}-d}\left(\|\mathcal{P}(\rho)\|_{\rm tr}-\frac{(d-1)(xM^{2}+d^{2})}{dM(M-1)}\right)$ based on $(8,2)$-POVM with $t=0.01$. The green curve is is the lower bound of ${\rm SN}(\rho)-1$ from Theorem 1 in Ref.\cite{ZWang} based on GSIC POVMs with $a=0.04984$. In other words, the green curve is the value of $\frac{d(d^{2}-1)}{ad^{3}-1}\left(\|\mathcal{P}(\rho)\|_{\rm tr}-\frac{ad^{2}+1}{d(d+1)}\right)$ based on GSIC POVM with $a=0.04984$. The purple curve is the lower bound of ${\rm SN}(\rho)-1$ from Result 1 in Ref.\cite{Tavakoli}. In other words, the purple curve is the value of $d(d+1)\left(\|\mathcal{P}(\rho)\|_{\rm tr}-\frac{2}{d(d+1)}\right)$ based on SIC POVM. The orange curve is the lower bound of ${\rm SN}(\rho)-1$ from realignment \cite{Hulpke,Johnston}. It is not difficult to see that Fig.3 is also a supplement to Fig.4 in Ref.\cite{Shang}. Obviously, our Schmidt number criterion detect more entangled states than other critera mentioned above. Thus, our criterion is more efficient in detecting the Schmidt number than other critera mentioned above.
    \begin{figure}[t]
    	\includegraphics[width=10cm]{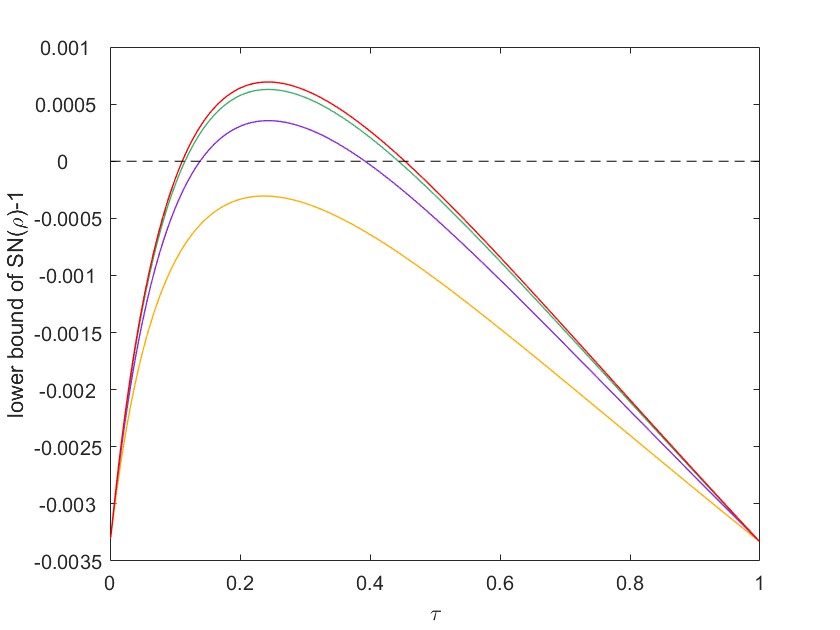}
    	\caption{Red curve is the lower bound of ${\rm SN}(\rho)-1$ from our Corollary 1 based on $(8,2)$-POVM with $t=0.01$. Green curve is is the lower bound of ${\rm SN}(\rho)-1$ from Theorem 1 in Ref.\cite{ZWang} based on GSIC POVMs with $a=0.04984$. Purple curve is the lower bound of ${\rm SN}(\rho)-1$ from Result 1 in Ref.\cite{Tavakoli}. Orange curve is the lower bound of ${\rm SN}(\rho)-1$ from realignment.}
    \end{figure}
\end{example}

\section{Conclusions and Discussions}
We have provided a criterion for detecting the Schmidt numbers of bipartite states based on symmetric measurements. As symmetric measurements cover GSIC POVM and MUBs, our criterion includes the GSIC and MUBs based criteria as particular cases. Through our examples, We have illustrated that our criterion is more effective than and superior to the GSIC criterion, the fidelity criterion and the realignment criterion. Moreover, we have obtained a class of symmetric measurement-induced lower bounds of concurrence for heterogeneous systems, which includes the one given in Ref.\cite{Haofan} as particular cases. It is possible that our results can be used to determining Schmidt numbers in experiments, which is consistent with the discussion in Ref.\cite{Tavakoli}. Our results may highlight further investigations on the Schmidt number criteria based on other quantum measurements or conical 2-designs. In addition, whether other norms can be used to obtain Schmidt-number criteria may also be a topic worth studying.

    \bigskip
    \noindent{\bf Acknowlegements}
This work is supported by the National Natural Science Foundation of China (NSFC) under Grant No. 12171044; the specific research fund of the Innovation Platform for Academicians of Hainan Province.

\begin{appendices}
\section{The Hermitian basis operators used to construct $(5,4)$-POVM  in Example 1}
In Example 1, we used $(5,4)$-POVM with the Hermitian basis operators $G_{\alpha, k}$ given by the following general Gell-Mann matrices \cite{Kalev1}:
$$G_{11}=\dfrac{1}{\sqrt{2}}\begin{pmatrix}
    	0 & -\mathrm{i} & 0 & 0 \\
    	\mathrm{i} & 0 & 0 & 0 \\
    	0 & 0 & 0 & 0 \\
    	0 & 0 & 0 & 0
    \end{pmatrix},~~~
G_{12}=\dfrac{1}{\sqrt{2}}\begin{pmatrix}
    	0 & 0 & -\mathrm{i} & 0 \\
    	0 & 0 & 0 & 0 \\
    	\mathrm{i} & 0 & 0 & 0 \\
    	0 & 0 & 0 & 0
    \end{pmatrix},~~~
G_{13}=\dfrac{1}{\sqrt{2}}\begin{pmatrix}
    	0 & 0 & 0 & -\mathrm{i} \\
    	0 & 0 & 0 & 0 \\
    	0 & 0 & 0 & 0 \\
    	\mathrm{i} & 0 & 0 & 0
    \end{pmatrix},
$$

$$G_{21}=\dfrac{1}{\sqrt{2}}\begin{pmatrix}
    	0 & 1 & 0 & 0 \\
    	1 & 0 & 0 & 0 \\
    	0 & 0 & 0 & 0 \\
    	0 & 0 & 0 & 0
    \end{pmatrix},~~~
    G_{22}=\dfrac{1}{\sqrt{2}}\begin{pmatrix}
    	0 & 0 & 0 & 0 \\
    	0 & 0 & -\mathrm{i} & 0 \\
    	0 & \mathrm{i} & 0 & 0 \\
    	0 & 0 & 0 & 0
    \end{pmatrix},~~~
    G_{23}=\dfrac{1}{\sqrt{2}}\begin{pmatrix}
    	0 & 0 & 0 & 0 \\
    	0 & 0 & 0 & -\mathrm{i} \\
    	0 & 0 & 0 & 0 \\
    	0 & \mathrm{i} & 0 & 0
    \end{pmatrix},
    $$

    $$G_{31}=\dfrac{1}{\sqrt{2}}\begin{pmatrix}
    	0 & 0 & 1 & 0 \\
    	0 & 0 & 0 & 0 \\
    	1 & 0 & 0 & 0 \\
    	0 & 0 & 0 & 0
    \end{pmatrix},~~~
    G_{32}=\dfrac{1}{\sqrt{2}}\begin{pmatrix}
    	0 & 0 & 0 & 0 \\
    	0 & 0 & 1 & 0 \\
    	0 & 1 & 0 & 0 \\
    	0 & 0 & 0 & 0
    \end{pmatrix},~~~
    G_{33}=\dfrac{1}{\sqrt{2}}\begin{pmatrix}
    	0 & 0 & 0 & 0 \\
    	0 & 0 & 0 & 0 \\
    	0 & 0 & 0 & -\mathrm{i} \\
    	0 & 0 & \mathrm{i} & 0
    \end{pmatrix},
    $$

    $$G_{41}=\dfrac{1}{\sqrt{2}}\begin{pmatrix}
    	0 & 0 & 0 & 1 \\
    	0 & 0 & 0 & 0 \\
    	0 & 0 & 0 & 0 \\
    	1 & 0 & 0 & 0
    \end{pmatrix},~~~
    G_{42}=\dfrac{1}{\sqrt{2}}\begin{pmatrix}
    	0 & 0 & 0 & 0 \\
    	0 & 0 & 0 & 1 \\
    	0 & 0 & 0 & 0 \\
    	0 & 1 & 0 & 0
    \end{pmatrix},~~~
    G_{43}=\dfrac{1}{\sqrt{2}}\begin{pmatrix}
    	0 & 0 & 0 & 0 \\
    	0 & 0 & 0 & 0 \\
    	0 & 0 & 0 & 1 \\
    	0 & 0 & 1 & 0
    \end{pmatrix},$$

    $$G_{51}=\dfrac{1}{\sqrt{2}}\begin{pmatrix}
    	1 & 0 & 0 & 0 \\
    	0 & -1 & 0 & 0 \\
    	0 & 0 & 0 & 0 \\
    	0 & 0 & 0 & 0
    \end{pmatrix},~~~
    G_{52}=\dfrac{1}{\sqrt{6}}\begin{pmatrix}
    	1 & 0 & 0 & 0 \\
    	0 & 1 & 0 & 0 \\
    	0 & 0 & -2 & 0 \\
    	0 & 0 & 0 & 0
    \end{pmatrix},~~~
    G_{53}=\dfrac{1}{2\sqrt{3}}\begin{pmatrix}
    	1 & 0 & 0 & 0 \\
    	0 & 1 & 0 & 0 \\
    	0 & 0 & 1 & 0 \\
    	0 & 0 & 0 & -3
    \end{pmatrix}.
    $$

\section{The Hermitian basis operators used to construct $(8,2)$-POVM in Example 3}
In Example 3, we used $(8,2)$-POVM with the Hermitian basis operators $G_{\alpha, k}$ given by the following Gell-Mann matrices,
    $$G_{11}=\dfrac{1}{\sqrt{2}}\begin{pmatrix}
    	0 & 1 & 0\\
    	1 & 0 & 0\\
    	0 & 0 & 0
    \end{pmatrix},~~~
    G_{21}=\dfrac{1}{\sqrt{2}}\begin{pmatrix}
    	0 & -\mathrm{i} & 0\\
    	\mathrm{i} & 0 & 0\\
    	0 & 0 & 0
    \end{pmatrix},~~~
    G_{31}=\dfrac{1}{\sqrt{2}}\begin{pmatrix}
    	0 & 0 & 1\\
    	0 & 0 & 0\\
    	1 & 0 & 0
    \end{pmatrix},
$$
$$G_{41}=\dfrac{1}{\sqrt{2}}\begin{pmatrix}
    	0 & 0 & -\mathrm{i}\\
    	0 & 0 & 0\\
    	\mathrm{i} & 0 & 0
    \end{pmatrix},~~~
    G_{51}=\dfrac{1}{\sqrt{2}}\begin{pmatrix}
    	0 & 0 & 0\\
    	0 & 0 & 1\\
    	0 & 1 & 0
    \end{pmatrix},~~~
    G_{61}=\dfrac{1}{\sqrt{2}}\begin{pmatrix}
    	0 & 0 & 0\\
    	0 & 0 & -\mathrm{i}\\
    	0 & \mathrm{i} & 0
    \end{pmatrix},
    $$
    $$G_{71}=\dfrac{1}{\sqrt{2}}\begin{pmatrix}
    	1 & 0 & 0\\
    	0 & -1 & 0\\
    	0 & 0 & 0
    \end{pmatrix},~~~
    G_{81}=\dfrac{1}{\sqrt{6}}\begin{pmatrix}
    	1 & 0 & 0\\
    	0 & 1 & 0\\
    	0 & 0 & -2
    \end{pmatrix}.
    $$

\end{appendices}
\end{document}